\documentclass[journal]{IEEEtran}%[12pt, draftclsnofoot, onecolumn]

\ifCLASSINFOpdf

\else

\fi

\usepackage{mathrsfs}
\usepackage{amsmath}
\usepackage{amsfonts}
\usepackage{amsthm}
\usepackage{amssymb}
\usepackage{graphicx}
\usepackage{subfigure}
\usepackage{indentfirst}
\usepackage{array}
\usepackage{epstopdf}
\usepackage{cite}
\usepackage{enumerate}
\usepackage{bm}
\usepackage[linesnumbered,ruled,vlined]{algorithm2e}
\usepackage{multirow}
\usepackage{verbatim}
\usepackage{stfloats}
\usepackage{color}
\newtheorem{theorem}{\bf{Theorem}}

\newtheorem{proposition}{\bf{Proposition}}

\newtheorem{definition}{\bf{Definition}}

\begin{document}

\title{\LARGE{Progressive Rate-Filling: A Framework for Agile Construction of Multilevel Polar-Coded Modulation}}

\author{Jincheng~Dai, Jinnan~Piao, and Kai~Niu

\thanks{This work was supported by the National Science Foundation of China (No. 62001049), the National Key Research and Development
Program of China (No. 2018YFE0205501), China Post-Doctoral Science Foundation (No. 2019M660032), and China Mobile Research Institute.}
\thanks{The authors are with the Key Laboratory of Universal Wireless Communications, Ministry of Education, Beijing University of Posts and Telecommunications (BUPT), Beijing 100876, China (email: daijincheng@bupt.edu.cn, piaojinnan@bupt.edu.cn, niukai@bupt.edu.cn).}
\vspace{-2em}
}

\maketitle

\begin{abstract}
In this letter, we propose a progressive rate-filling method as a framework to
study agile construction of multilevel polar-coded modulation. We show that the bit indices within each component polar code can follow a fixed, precomputed ranking sequence, e.g., the Polar sequence in the 5G standard, while their allocated rates (i.e., the number of information bits of each component polar code) can be fast computed by exploiting the target sum-rate approximation and proper rate-filling methods. In particular, we develop two rate-filling strategies based on the capacity and the rate considering the finite block-length effect. The proposed construction methods can be performed independently of the actual channel condition with ${O\left(m\right)}$ ($m$ denotes the modulation order) complexity and robust to diverse modulation and coding schemes in the 5G standard, which is a desired feature for practical systems.
\end{abstract}

\begin{IEEEkeywords}
Polar-coded modulation, code construction, progressive rate-filling, target sum-rate approximation.
\end{IEEEkeywords}

\IEEEpeerreviewmaketitle

\section{Introduction}\label{section_introduction}

\IEEEPARstart{P}{olar} codes, invented by Ar{\i}kan \cite{arikan}, show competitive performance compared to state-of-the-art error-correcting codes for a wide
range of code lengths and have been adopted in 5G new radio (NR) \cite{5G_NR_std}. Shortly after polar codes were invented, the underlying channel polarization phenomenon has been found universal in many other signal processing problems \cite{Polar_coded_modulation_seidl,Polar_coded_MIMO_dai}. Particularly, in order to improve spectral efficiency, Seidl has introduced a $2^m$-ary polar-coded modulation scheme \cite{Polar_coded_modulation_seidl}. Taking the dependencies among the bits that are mapped to a modulation symbol as a special kind of channel transform, polar-coded modulation (PCM) is derived under the framework of two-stage channel polarization, i.e., the modulation partition and the binary partition. The whole structure of PCM follows the multilevel coding (MLC) scheme \cite{MLC_wacha}, however, the novel two-stage channel transform concatenated manner accommodates a joint design of polar coding and modulation which allows one to describe the two techniques in a unified context rather than a simple combination.

In practical wireless communication systems, adaptive modulation and coding (AMC) is one of the most critical link adaption techniques \cite{Tse,5G_NR_std_214}. AMC works by dynamically adjusting the modulation and coding scheme (MCS), including modulation order and coding rate, to align with the instantaneous channel quality, and thus the link throughput can be maximized \cite{MIESM}. Note that a critical issue in practical coded modulation systems is the code construction, which is expected to be agile and robust to diverse modulation and coding schemes. Similar to that of standard polar codes, the construction of an MLC-PCM system can also be implemented through density evolution (DE) \cite{DE_mori} or Gaussian approximation (GA) \cite{GA} algorithms. The reliabilities of all bit polarized sub-channels belonging to all component codes are calculated and sorted to select the information indices. Despite their high accuracy, these online construction methods usually rely on specific channel parameters and involve high computational complexity that scales linearly with the block-length. Worse still, these methods also introduce sorting complexity. Recently, the polarization weight (PW) sequence \cite{PW} provides a new idea to construct polar codes without dependence on transmission channels. Inspired by the PW sequence, the Polar sequence \cite{5G_NR_std} has been employed for polar code construction in the 5G standard. However, the PW/Polar sequence can only be used for index-ranking within each component code. To take advantage of the PW/Polar sequence in the MLC-PCM system, the first critical step is to accurately allocate the number of information bits for each component code. The related works and efficient solutions are still missing. Furthermore, the high computational and sorting complexities incurred by DE/GA online construction methods \cite{PCM_construction_wksp} depending on channel conditions are highly impractical for AMC systems with varying parameters. The expected construction scheme is \emph{universal}, in the sense that it can be performed on-the-fly and efficiently at ideal $O\left(1\right)$ cost as in 5G NR, also it should be robust to any channel parameter in the AMC system.

In this letter, we aim to set up a new framework for agile and robust construction of MLC-PCM. We analyze the properties of MLC-PCM structure and show that the bit indices within each component polar code can follow a fixed, precomputed ranking
sequence, e.g., the PW sequence or the Polar sequence in the 5G standard, while their allocated rates (i.e., the number of information bits of each component polar code) can be fast computed with the target sum-rate approximation and proper \emph{rate-filling} methods. The proposed code construction methods need only an one-shot operation with $O\left(m\right)$ complexity, which provides a answer to the MLC-PCM construction question of practical interest. In particular, we propose two rate-filling strategies. One is based on the bit channel capacities under the first-stage modulation partition, where the reference signal-to-noise-ratio (SNR) is autonomously obtained by finding out an equivalent channel whose capacity is equal to the target sum-rate of total $m$ component polar codes. By this means, the code construction is realized agilely for arbitrary MCS and independently of the actual channel condition. Also, it is proven to be capacity-achieving as the block-length goes to infinity. The other rate-filling strategy considers the rate with the finite block-length effect \cite{capacity_finite}, which can further improve the rate-filling accuracy by slightly increasing computational complexity. The proposed rate-filling methods are carried out in a progressive manner which ensures the summation of allocated information bits of all component polar codes strictly equal to the total number of information bits.

%Simulation results show that the proposed construction methods are robust to varying parameters and can achieve stable performance gain compared to the LDPC-coded modulation schemes in the 5G standard \cite{5G_NR_std}.

\emph{Notational Conventions:} In this letter, the calligraphic characters, such as ${\mathcal X}$, are used to denote sets. Let $\left| \mathcal X \right|$ denote the cardinality of the set $\mathcal X$. We write lowercase letters (e.g., $x$) to denote scalars. The notation ${v}_1^N$ denotes a $N$-dimensional row vector $\left( {{v_1},{v_2}, \cdots ,{v_N}} \right)$, we use ${v}_i^j$ to denote a subvector $\left( {{v_i},{v_{i+1}}, \cdots ,{v_j}} \right)$ of ${v}_1^N$, and $v_i$ denotes the $i$-th element in ${ v}_1^N$. The bold letters, such as ${\bf{X}}$, stand for matrices. Specially, for positive integer $N$, $\left[\kern-0.15em\left[ N \right]\kern-0.15em\right] \triangleq \left\{ {1,2, \cdots ,N} \right\}$. The function $\left\lceil \cdot \right\rceil$ denotes the ceil operation.

\section{Multilevel Polar-Coded Modulation}\label{section_pcm}

Let ${ W}: {\mathcal X} \to {\mathcal Y}$ denote a discrete memoryless channel (DMC) with input modulation symbol $x \in {\mathcal X}$ (alphabet size $\left| {\mathcal X} \right| = 2^m$), output symbols $y \in {\mathcal Y}$ from an arbitrary alphabet $\mathcal Y$, and the mutual information is marked as $I\left( {X;Y} \right)$. Given the $m$-bit sequence $b_1^m$ to modulation symbol mapping rule $\varphi : {\left\{ 0,1\right\}}^{m}  \mapsto {\mathcal X}$, the order-$m$ \emph{modulation partition (MP)} can be characterized as ${\varphi}: {W} \to \left( {{{W}_1},{{W}_2}, \cdots ,{{W}_m}} \right)$. It transforms $W$ to an ordered set of binary-input memoryless channels (BMCs) ${{W}_{{k}}}$ with $k=1,2,\cdots,m$, and they are referred as the \emph{bit subchannels}. Each ${W}_k$ of an $m$-MP is assumed to have the knowledge of the output of $W$ and the transmitted bits over the bit subchannels ${ W}_{k'}$ with smaller indices $k' = 1,2,\cdots,k-1$, i.e.,
\begin{equation}\label{trans_prob_asy_PCM}
\begin{aligned}
{{W}_k}\left( {y,b_1^{k - 1}\left| {{b_k}} \right.} \right) = \frac{1}{{{2^{m - 1}}}}  \sum\limits_{b_{k + 1}^m \in {{\left\{ {0,1} \right\}}^{m - k}}} { { {W}\left( {y\left| {x = \varphi\left( {b_1^m} \right)} \right.} \right)} }.
\end{aligned}
\end{equation}
Hence, from the perspective of mutual information, we have
\begin{equation}\label{mutual_information}
\sum\limits_{k = 1}^m {I\left( {{W_k}} \right)} = \sum\limits_{k = 1}^m {I\left( {{B_k};Y\left| {{B_1},{B_2}, \cdots ,{B_{k - 1}}} \right.} \right)} = I\left( {X;Y} \right),
\end{equation}
where $B_k$ denotes the random variable corresponding to $b_k$ in $b_1^m$. Accordingly, the MP preserves the capacity.

The second stage performs the conventional $N$-dimensional binary channel polarization transform ${\bf{G}}_N = {{\bf F}_2^{ \otimes n}}$ \cite{arikan,5G_NR_std} on each of these $m$ bit synthesized subchannels $W_k$, where $N = 2^n, n=1,2,\cdots$, ${{\bf{F}}_2} = \left[ { \begin{smallmatrix} 1 & 0 \\ 1 &  1 \end{smallmatrix} } \right]$, and ${\bf F}_2^{ \otimes n}$ stands for the $n$-th Kronecker power of ${{\bf{F}}_2}$. The resulting BMCs $\{ {W_{k}^{\left( i \right)}} \}$ with $i=1,2,\cdots,N$ and $k=1,2,\cdots,m$ are referred as the \emph{bit polarized subchannels}. The two-stage channel transform is shown in Fig.~\ref{channel_transform_figure}. Followed by \cite{arikan}, the transition probability of each bit polarized subchannel is written as
\begin{equation}\label{trans_prob_stage2}
  \begin{aligned}
& W_{k}^{\left( i \right)}\left( {y_1^N,u_1^{a - 1}\left| {{u_a}} \right.} \right) \\
& = \sum\limits_{u_{a + 1}^{kN} \in {{\left\{ {0,1} \right\}}^{N - i}}} {\left( {\frac{1}{{{2^{N - 1}}}}\prod\limits_{i' = 1}^N {{W_k}\left( {{y_{i'}},b_1^{k - 1}\left( {i'} \right)\left| {{b_k}\left( {i'} \right)} \right.} \right)} } \right)},
\end{aligned}
\end{equation}
where $a = \left( {k - 1} \right)N + i$, the bit vector $b_1^k\left( {i'} \right)$ corresponding to the $i'$-th symbol is $b_1^k\left( {i'} \right) \triangleq ( {{v_{i'}},{v_{i' + N}}, \cdots ,{v_{i' + \left( {m - 1} \right)N}}} )$, and $v_{\left( {k - 1} \right)N + 1}^{kN} = u_{\left( {k - 1} \right)N + 1}^{kN}{{\mathbf G}_N}$.

\begin{figure}[t]
\setlength{\abovecaptionskip}{0.cm}
\setlength{\belowcaptionskip}{-0.cm}
  \centering{\includegraphics[scale=0.89]{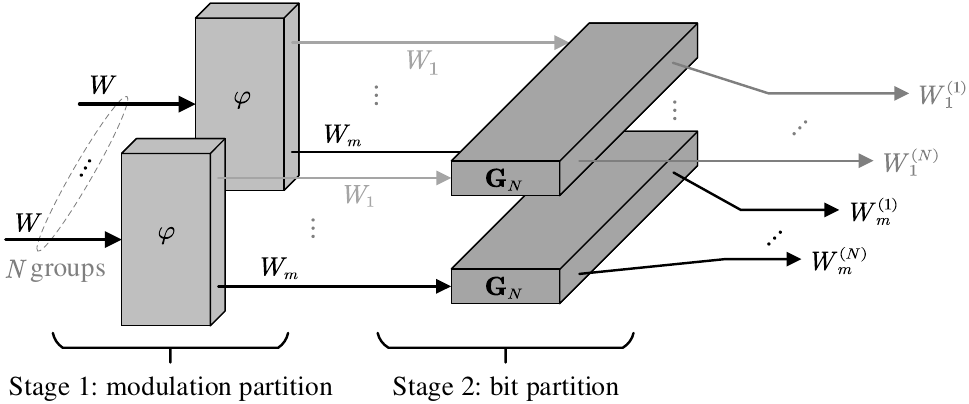}}
  \caption{The two-stage channel transform in MLC-PCM.}\label{channel_transform_figure}
\vspace{-1em}
\end{figure}

In order to construct a MLC-PCM system of rate $R$ and symbol block-length $N$ (i.e., the block-length of each component polar code) over the channel $W$, the $K$ most reliable indices among the $mN$ bit polarized subchannels are selected as the \emph{information} indices. The information indices corresponding to the $k$-th component polar code is denoted by ${\mathcal A}_{k}$ and the others are \emph{frozen} bits known by the transceiver, marked as ${{\mathcal A}_{k}^c} = \left[\kern-0.15em\left[ N \right]\kern-0.15em\right]\backslash {{\mathcal A}_{k}}$. The total number of information indices is $K =\sum\nolimits_{k = 1}^m {\left|{{\mathcal A}_{k}}\right|}$. Intuitively, the second stage channel transform can be viewed as constructing a set of $N$-length component polar codes on each of the $m$ bit subchannels in the set $\{ W_k \}$ with a total coding rate $R = \frac{K}{mN}$.

The widely-used code construction method of MLC-PCM is based on the DE/GA algorithm \cite{Polar_coded_MIMO_dai,MLC_wacha}. Given the channel parameters, it first calculates the capacities of bit subchannels under the first-stage MP, i.e.,
\begin{equation}\label{eq_capacity_W_syn}
  I\left( {{{ W}_k}} \right) = \sum\limits_{{b}_1^k \in {{\left\{ {0,1} \right\}}^k}} {\sum\limits_{y \in {\mathcal Y}} {\frac{ \Pr \left( {y\left| {{ b}_1^k} \right.} \right)}{{{2^k}}} \cdot {{\log }_2}\frac{{\Pr \left( {y\left| {{ b}_1^k} \right.} \right)}}{{\Pr \left( {y\left| {{ b}_1^{k - 1}} \right.} \right)}}} },
\end{equation}
where the transition probability is
\begin{equation}\label{eq_trans_prob}
  \Pr \left( {y\left| {{ b}_1^k} \right.} \right) = \frac{1}{{{2^{m - k}}}}\sum\limits_{{ b}_{k + 1}^m \in {{\left\{ {0,1} \right\}}^{m - k}}} { W\left( {y\left| {x = \varphi \left( {{ b}_1^m} \right)} \right.} \right)}.
\end{equation}
Then, each $W_k$ is approximated by a binary-input additive white Gaussian noise (BI-AWGN) channel ${W}_k'$ whose capacity is equal to $I\left( {{{ W}_k}} \right)$. The reliabilities of $mN$ bit polarized subchannels are computed by using surrogate channels $\{{W}_k'\}$ and the DE/GA algorithm \cite{Polar_coded_MIMO_dai,PCM_construction_wksp}. The information indices $\{{\mathcal A}_{k}\}$ are obtained by ranking all these bit polarized subchannels with their reliabilities.

%This online construction relies on the actual channel condition, whose main drawback is their high computational and sorting complexity and thus unacceptable for practical AMC systems

%with varying parameters such as modulation order $m$, coding rate $R$, and symbol block-length $N$.

\section{Progressive Rate-Filling Construction}\label{section_rate_filling}

In this section, we propose two rate-filling strategies for the fast and low-complexity construction of MLC-PCM. Note that each component polar code corresponds to one bit subchannel $W_k$, and thus their $N$ bit indices can be sorted by using the PW sequence or the Polar sequence in the 5G standard. There is no need to sort all the $mN$ bit polarized subchannels like that in the online DE/GA construction, and thus the main task is to compute their allocated rates, i.e., the number of information bits of each component polar code $\left| {{{\mathcal A}_{k}}} \right|$.

\subsection{Rate-Filling Based on Capacity}\label{subsection_RFI}

To find out the number of allocated information bits $K_k = \left| {{{\mathcal A}_{k}}} \right|$ for each component polar code, we propose a rate-filling method with $R_k = K_k/N$. Note that we know the total number of information bits $K$, the target transmission rate is $R_T \triangleq K/N = mR$ bits per symbol, which is also referred as the sum-rate of component polar codes. Hence, we can adjust the parameters of $W$, e.g., the SNR of $W$ under the AWGN channel, to find an equivalent channel $\overline W$ whose capacity is equal to the target sum-rate $R_T$, i.e.,
\begin{equation}\label{capacity_equivalent}
\begin{aligned}
 & {R_T} = I\left( {\overline W} \right) \Rightarrow \\
 & mR = \sum\limits_{x \in {\mathcal X}} {\sum\limits_{y \in {\mathcal Y}} {\frac{1}{{{2^m}}} \cdot \overline W\left( {y\left| x \right.} \right) \cdot {{\log }_2}\frac{{\overline W\left( {y\left| x \right.} \right)}}{{\sum\limits_{x' \in {\mathcal X}} {\frac{1}{{{2^m}}} \cdot \overline W} \left( {y\left| x' \right.} \right)}}} }.
\end{aligned}
\end{equation}
For this identical sum-rate approximated channel $\overline W$, we can also compute the bit subchannel capacities $\{ I\left( {{{\overline W}_k}} \right) \}$ under the MP as \eqref{eq_capacity_W_syn}.

\begin{proposition}\label{prop_capacity}
\emph{
  The rate-filling for each component polar code is performed as
  \begin{equation}\label{eq_prop_rate_filling_capacity}
    {R_k} \leftarrow mR \cdot \frac{{I\left( {{{\overline W}_k}} \right)}}{{I\left( {{{\overline W}}} \right)}} = {I\left( {{{\overline W}_k}} \right)}
  \end{equation}
  for $k = 1,2,\cdots,m$.
  }
\end{proposition}

A special note is that the code construction is intended for polar codes, which means all the source bits along with the attached bits, e.g., the cyclic-redundancy-check (CRC) bits \cite{CASCL_niukai} etc., should be counted in the number of information bits $K$. In practical AMC systems, this rate-filling is independent with block-length $N$, which enables the code construction to be robust to varying parameters. The channel condition, e.g., SNR, determines the code construction in an implicit way that selects the MCS, including $m$ and $R$, to align with the instantaneous channel state in AMC systems.

Next, we prove that the proposed rate-filling construction scheme is capacity-achieving. First, we give the definition of \emph{consistence property} for modulation channels.

\begin{definition}\label{definiton}
\emph{
  Given two modulation channels $W'$ and $W''$, suppose $W'$ is upgraded with respect to $W''$ \cite{rateless_libin}, i.e., $W' \succeq W''$ ($I(W') \ge I( {W''} )$), if we can derive that their split bit subchannels satisfy $W'_k \succeq W''_k$ for any $k=1,2,\cdots,m$, we define this as the \emph{consistence property}.
  }
\end{definition}

For example, given the modulation scheme and modulation order $m$, two AWGN modulation channels with different SNRs satisfy the consistence property.

\begin{theorem}
  For any modulation channel satisfying the consistence property, the proposed rate-filling strategy is capacity-achieving as the block-length $N$ goes to infinity.
\end{theorem}

\begin{proof}
  Applying the channel polarization theorem \cite{arikan} to MLC-PCM, for each $N$-length component polar code, the subchannels $\{ {W_{k,N}^{\left( i \right)}} \}$ polarize in the sense that, for any fix $\delta \in (0,1)$, as $N$ goes to infinity, the fraction of indices $i \in \left[\kern-0.15em\left[ N
 \right]\kern-0.15em\right]$ for which $I( {W_{k}^{\left( i \right)}} ) \in (1-\delta,1]$ goes to $I \left(W_k\right)$, and the fraction for which $I( {W_{k}^{\left( i \right)}} ) \in [0,\delta)$ goes to $1 - I \left(W_k\right)$. The good bit indices set ${\mathcal S}(W_k)$ for $W_k$ (i.e., those with systematic capacity near 1) satisfy $\mathop {\lim }\limits_{N \to  + \infty } \frac{{\left| {\mathcal S}(W_k) \right|}}{N} = I\left( {{W_k}} \right)$.

Given the total coding rate $R$ and the modulation scheme with order $m$ in MLC-PCM, the proposed rate-filling construction is performed as \eqref{eq_prop_rate_filling_capacity}. Under this modulation configuration, it is sufficient to prove that for any channel $W$ which is upgraded with respect to $\overline W$ in \eqref{capacity_equivalent}, i.e., $W \succeq {\overline W}$ ($I(W) \ge I( {\overline W} ) = mR$), suppose $W$ and ${\overline W}$ satisfy the consistence property, we have the system error probability $\epsilon \to 0$ as $N \to +\infty$. Since $W_k \succeq {\overline W}_k$ for any $k = 1,2,\cdots,m$, due to the \emph{nesting property} of polar codes \cite{rateless_libin}, it can be derived that good bit indices ${\mathcal S}({\overline W}_k)$ for ${\overline W}_k$ (i.e., those with systematic capacity near 1) must be a subset of good bit indices ${\mathcal S}(W_k)$ for $W_k$, i.e., ${\mathcal S}({\overline W}_k) \subset {\mathcal S}(W_k)$. Note that ${\mathcal S}({\overline W}_k)$ is indeed the selected information set ${\mathcal A}_k$ as $N$ goes to infinity, and thus all the information-carrying subchannels are of capacity 1 such that the system error probability $\epsilon \to 0$. In this sense, the capacity-achieving property of the proposed rate-filling construction is proven.
\end{proof}

In practical implementation, given the modulation and coding parameters, the proposed rate-filling scheme should strictly ensure the summation of allocated information bits of each component polar code equal to the total number of information bits, i.e., $\sum\nolimits_{k = 1}^m {{K_k}}  = K$. To this end, the proposed rate-filling is performed in a progressive manner which operates from the bit subchannel of the highest capacity to the one of the lowest capacity. The whole procedure is summarized as Algorithm \ref{rate_filling_algorithm}.

\begin{algorithm}[t]

\caption{Progressive rate-filling code construction}\label{rate_filling_algorithm}

Initialize the system parameters: modulation order $m$, symbol length (component code length) $N$, the total number of information bits $K$, symbol mapping rule $\varphi : {\left\{ 0,1\right\}}^{m}  \mapsto {\mathcal X}$\;

Find a channel $\overline W$ with the identical sum-rate as \eqref{capacity_equivalent}\;

Calculate the capacity of each bit subchannel $\{I({\overline W}_k)\}$\;

Sort the bit synthesized subchannels such that $I({\overline W}_{k_1}) \ge I({\overline W}_{k_2}) \ge \cdots \ge I({\overline W}_{k_m})$\;

\For{$t = 1,2,\cdots,m$}
{
    Compute the number of information bits for the $k_t$-th component polar code as
    \begin{equation*}
      {K_{{k_t}}} = \left\lceil \left( {K - \sum\nolimits_{t' = 1}^{t-1} {{K_{{k_{t'}}}}} } \right) \cdot \frac{{I\left( {{{\overline W}_{{k_t}}}} \right)}}{{\sum\nolimits_{t' = t}^m {I\left( {{{\overline W}_{{k_{t'}}}}} \right)} }} \right\rceil.
    \end{equation*}\\
    Determine the information set ${\mathcal A}_{k_t}$ according to $K_{k_t}$ and the Polar sequence in \cite[Table 5.3.1.2-1]{5G_NR_std}\;
}
		
\textbf{return} The information sets of each component polar code ${\mathcal A}_{1},{\mathcal A}_{2},\cdots,{\mathcal A}_{m}$.

\end{algorithm}

\subsection{Rate-Filling Based on Finite Block-Length Rate}\label{subsection_RFII}

Note that the above rate-filling scheme based on capacity does not take into account the block-length $N$ of each component polar code, this can well match most cases especially for large $N$. However, for some cases, the accuracy of rate-filling can be further improved by considering the rate with the finite block-length effect \cite{capacity_finite}. For the $k$-th component code, given the block-length $N$, the maximal rate achievable with error probability ${\epsilon}_k$ is closely approximated as
\begin{equation}\label{eq_finite_cap}
  M\left( {{W_k},N,{\epsilon_k}} \right) = I\left( {{W_k}} \right) - \sqrt {\frac{{{V_k}}}{N}} {Q^{ - 1}}\left( {{\epsilon_k}} \right),
\end{equation}
where $Q$ is the complementary Gaussian cumulative distribution function, and the channel dispersion $V_k$ is computed as
\begin{equation}
\begin{aligned}
  & V_k = \\
  & \sum\limits_{{b}_1^k \in {{\left\{ {0,1} \right\}}^k}} {\sum\limits_{y \in {\mathcal Y}} {\frac{\Pr \left( {y\left| {{ b}_1^k} \right.} \right)}{{{2^k}}} \cdot \left[{{\log }_2}\frac{{\Pr \left( {y\left| {{ b}_1^k} \right.} \right)}}{{\Pr \left( {y\left| {{ b}_1^{k - 1}} \right.} \right)}}\right]^2} } - I^2\left( {{{ W}_k}} \right),
\end{aligned}
\end{equation}
where $I\left( {{{ W}_k}} \right)$ and the transition probabilities are defined in \eqref{eq_capacity_W_syn} and \eqref{eq_trans_prob}. Given the target system error probability $\epsilon$, it is still difficult to determine the error probability $\epsilon_k$ of each component code. However, since the rate-filling process may allocate more information bits to component codes of higher capacity, the error probabilities $\epsilon_k$ tend to be close with each other. Therefore, we can approximate $\epsilon_k = {\overline \epsilon}$ for any $k=1,2,\cdots,m$, and given the target block-error ratio (BLER) $\epsilon$, we can derive that
\begin{equation}\label{eq_error_prob}
\begin{aligned}
  \epsilon = 1 - \prod\limits_{k = 1}^m {\left( {1 - {\epsilon_k}} \right)}  & \Rightarrow 1 - \epsilon = {\left( {1 - {\overline \epsilon}} \right)^m} \\
  ~ & \Rightarrow \epsilon_k = {\overline \epsilon} = 1 - {\left( {1 - \epsilon} \right)^{\frac{1}{m}}}.
\end{aligned}
\end{equation}
In this way, given the system target BLER $\epsilon$, one can compute the error
probability $\epsilon_k$ for each component code. To allocate the coding rates for MLC-PCM, we first adjust the parameters of $W$ to find a channel $\widetilde W$ whose achievable sum-rate at finite block-length $N$ is equal to $R_T$, i.e.,
\begin{equation}
  {R_T} = mR =  \sum\limits_{k = 1}^m {M\left( {{{\widetilde W}_k},N,{\epsilon_k}} \right)},
\end{equation}
where ${M( {{{\widetilde W}_k},N,{\epsilon_k}} )}$ is calculated as \eqref{eq_finite_cap} and the target BLER $\epsilon$ can be chosen with the practical requirements, e.g., $10^{-1}$ \cite{MIESM}, etc., such that $\epsilon_k$ is computed as \eqref{eq_error_prob}. By using the surrogate channel $\widetilde W$, the rate-filling procedure is given as \eqref{eq_prop_rate_filling_rate_finite}.

\begin{proposition}\label{prop_rate}
\emph{
  The rate-filling for each component polar code is performed as
  \begin{equation}\label{eq_prop_rate_filling_rate_finite}
    {R_k} \leftarrow mR \cdot \frac{{M\left( {{{\widetilde W}_k},N,{\epsilon_k}} \right)}}{\sum\limits_{k' = 1}^m {M\left( {{{\widetilde W}_{k'}},N,{\epsilon_{k'}}} \right)}} = {M\left( {{{\widetilde W}_k},N,{\epsilon_k}} \right)}
  \end{equation}
  for $k = 1,2,\cdots,m$.
  }
\end{proposition}

Apparently, the proposed rate-filling scheme in \emph{Proposition \ref{prop_rate}} is an enhanced version of that in \emph{Proposition \ref{prop_capacity}}. It improves the accuracy of rate-filling with slightly increased computational complexity in \eqref{eq_finite_cap}. In practical implementation, the rate-filling should also be carried out in a progressive manner, and the whole procedure is similar to that in Algorithm \ref{rate_filling_algorithm} where the capacity terms $\{I( {{{\overline W}_k}} ) \}$ are replaced with the finite block-length rate terms $\{ {M( {{{\widetilde W}_k},N,{\epsilon_k}} )} \}$.

\section{Performance Evaluation}\label{section_performance}

We make a comprehensive performance evaluation of the proposed rate-filling code construction methods. The performance of MLC-PCM systems constructed by state-of-the-art GA methods \cite{Polar_coded_MIMO_dai} (i.e., the ``LM-DGA'' method in \cite{PCM_construction_wksp}) and that of the LDPC-coded modulation in the 5G standard \cite{5G_NR_std} are also provided as comparisons. For ease of exposition, the proposed rate-filling method based on capacity in Subsection.\ref{subsection_RFI} is abbreviated as ``RF-I'', and the proposed rate-filling method based on the finite block-length rate in Subsection.\ref{subsection_RFII} is abbreviated as ``RF-II''. In addition, the target BLER $\epsilon$ in \eqref{eq_error_prob} of the RF-II construction method is set to $10^{-1}$ for all the sequent simulation cases, which follows the practical BLER requirement in 4G and 5G wireless systems \cite{MIESM}. The CRC-aided successive cancellation list (CA-SCL) decoding \cite{CASCL_niukai} is used for MLC-PCM, where the list size is 32 and the 16-bit CRC sequence in the 5G standard \cite{5G_NR_std} is adopted. The sum-product algorithm (SPA) with 20 iterations and layered scheduling is used for LDPC decoding \cite{linshu}, which presents comparable complexity with polar decoding \cite{Polar_coded_MIMO_dai}.

\begin{figure}[t]
\setlength{\abovecaptionskip}{0.cm}
\setlength{\belowcaptionskip}{-0.cm}
  \centering{\includegraphics[scale=0.58]{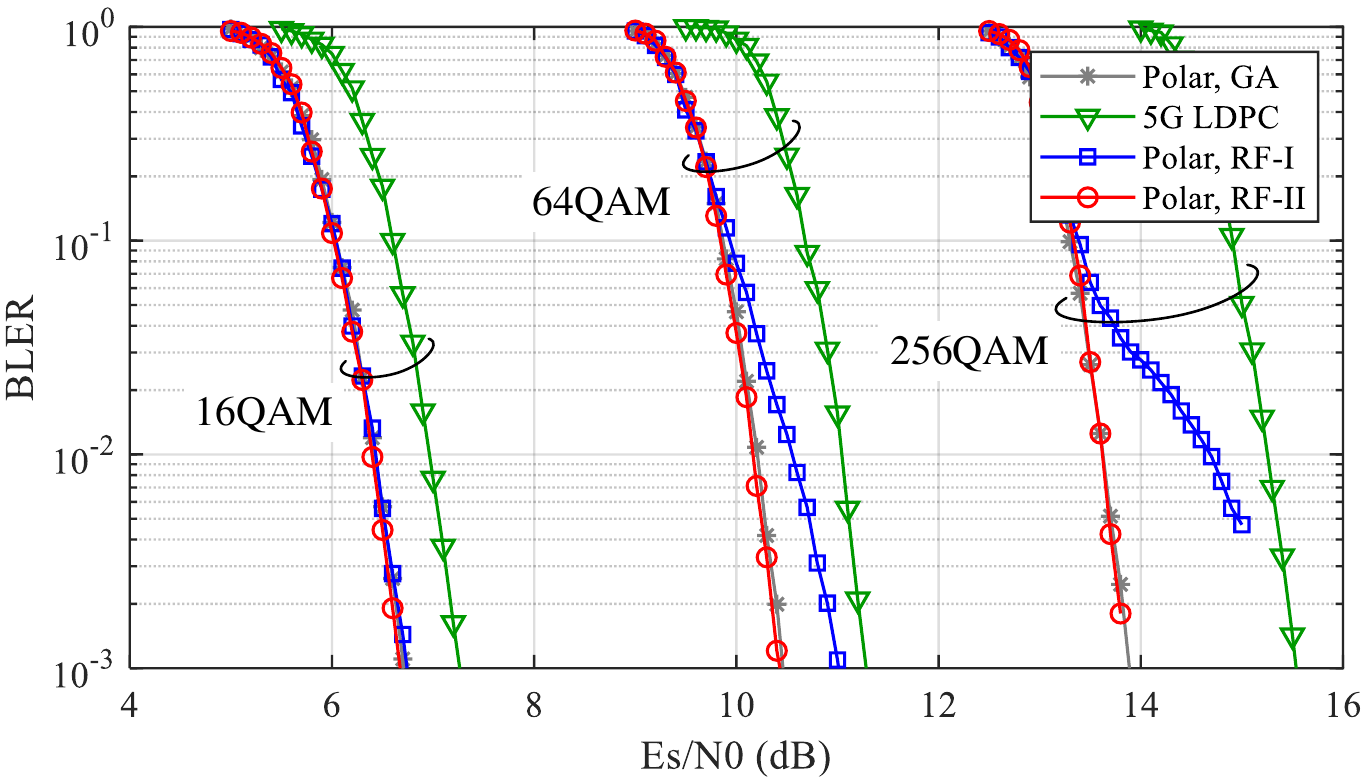}}
  \caption{BLER performance comparisons with $R = 0.5$ and $N = 512$.}\label{BLER_figure}
\vspace{-1em}
\end{figure}

\begin{figure*}[t]
\setlength{\abovecaptionskip}{0.cm}
\setlength{\belowcaptionskip}{-0.cm}
  \centering{\includegraphics[scale=0.45]{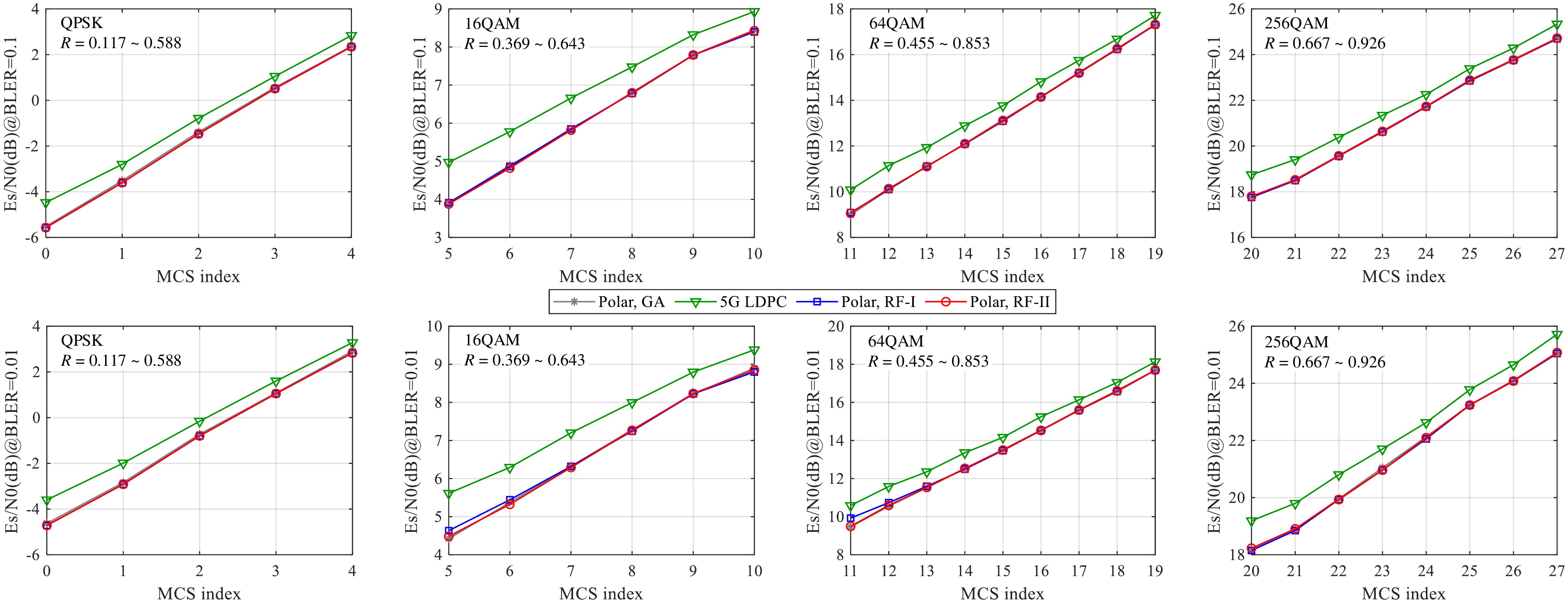}}
  \caption{The minimum required SNR to achieve ${\text{BLER}} = 10^{-1}~{\text{or}}~10^{-2}$ under the AWGN channel, where the symbol block-length $N=256$ and the MCS follows the 5G standard \cite[Table 5.1.3.1-2]{5G_NR_std_214}.}\label{MCS_figure}
\vspace{-1em}
\end{figure*}

In Fig. \ref{BLER_figure}, we compare the BLER performance of various coded modulation schemes under the AWGN channel. For all cases, the code rate is $R = 0.5$ and the symbol block-length is set to $N = 512$ which corresponds to component block-length in the MLC-PCM. The modulation scheme employs the quadrature amplitude modulation (QAM) \cite{5G_NR_std}. Given the modulation order of QAM $m \in \{4,6,8\}$, the block-length of LDPC code is $mN$. It should be noted that the online GA construction in MLC-PCM is executed individually for each evaluated SNR value, but the proposed rate-filling construction method is an one-shot operation that is independent with the evaluated SNR.

Clearly, the MLC-PCM constructed by RF-I can achieve almost identical performance as the GA construction under 16QAM. With modulation order increasing, the performance of RF-I becomes worse. That is because the number of component polar codes increase with the modulation order, and thus the accuracy of RF-I without considering the block-length decreases for high-order modulation schemes, e.g., 64QAM and 256QAM, etc. However, the performance of MLC-PCM constructed by RF-II aligns well with the GA construction for diverse modulation orders, which validates the effect of considering finite block-length. Therefore, the proposed RF-II construction method is more robust with a slight increase of computation complexity in \eqref{eq_finite_cap}. Compared to the LDPC-coded modulation in the 5G standard \cite{5G_NR_std}, the MLC-PCM shows stable performance gain.

In Fig. \ref{MCS_figure}, we provide the minimum required SNR to achieve ${\text{BLER}} = 10^{-1}~{\text{or}}~10^{-2}$ under the AWGN channel with the symbol block-length $N = 256$. The employed MCS (including modulation order $m$ and code rate $R$) follows the 5G standard \cite[Table 5.1.3.1-2]{5G_NR_std_214}, and the corresponding code rate range is marked in each subfigure. When ${\text{BLER}} = 10^{-1}$, both the RF-I and the RF-II align well with the GA construction, which outperforms the 5G LDPC-coded modulation schemes. When ${\text{BLER}} = 10^{-2}$, the MLC-PCM constructed by the RF-I shows some tiny loss at MCS-5 and MCS-11 while the RF-II still performs well for all cases. Note that even for high modulation orders, the RF-I results under the MCS in the 5G standard \cite{5G_NR_std_214} do not show obvious loss like that in Fig. \ref{BLER_figure}. That indicates the robustness of the proposed rate-filling methods for practical configurations as the 5G standard.

\begin{figure}[t]
\setlength{\abovecaptionskip}{0.cm}
\setlength{\belowcaptionskip}{-0.cm}
  \centering{\includegraphics[scale=0.58]{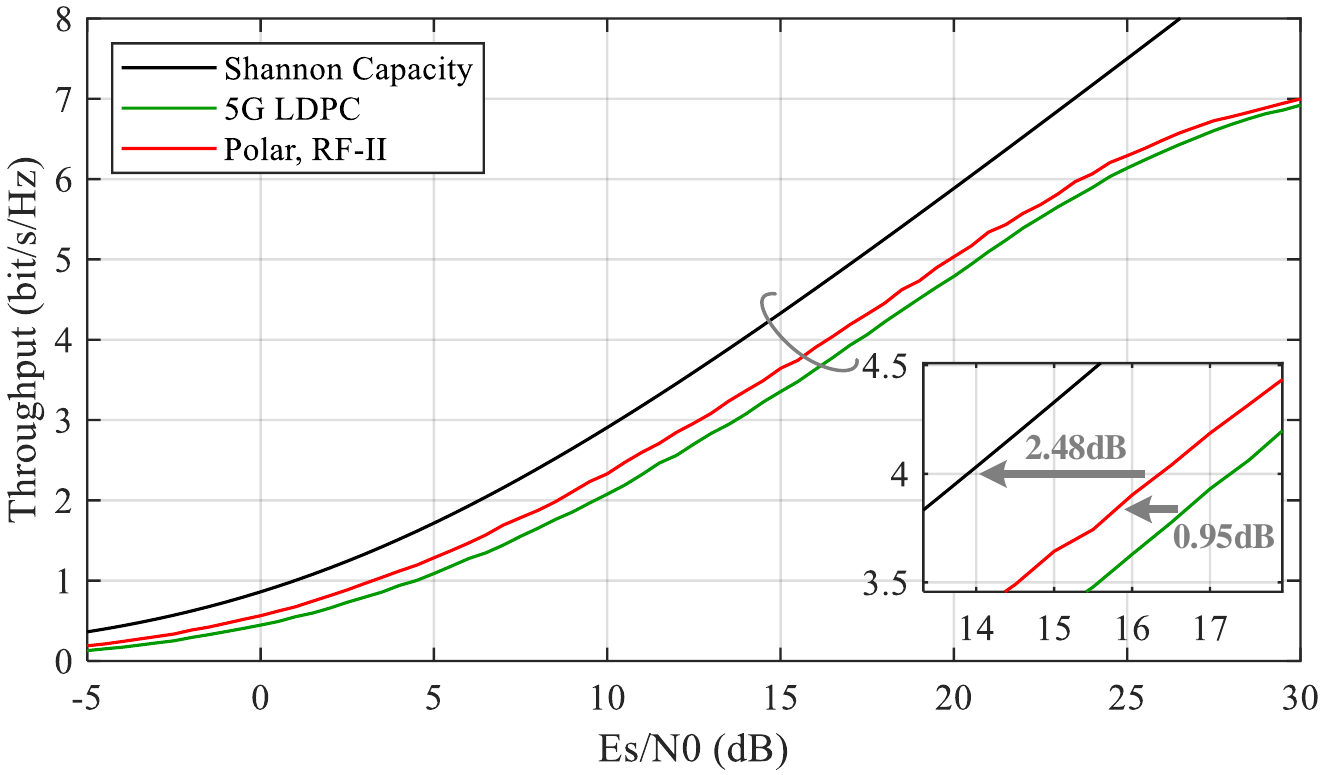}}
  \caption{Link throughput comparison under the block-fading channel.}\label{TP_figure}
\vspace{-0em}
\end{figure}

Fig. \ref{TP_figure} shows the link throughput comparison result under the block-fading channel, where the symbol block-length is $N = 256$. The AMC mechanism adaptively selects MCS to maximize the link throughput under BLER constraint $\le 10^{-1}$ \cite{MIESM}, which is indeed employed in the 5G NR system. Clearly, the MLC-PCM constructed by the proposed RF-II method shows stable gain with respect to the LDPC-coded modulation scheme in the 5G standard \cite{5G_NR_std}. That validates the agility and robustness of the proposed code construction method.

Regarding the construction complexity, it is straightforward to see that each component code is of $O\left(1\right)$ complexity as in 5G NR so that the whole progressive procedure involves $O\left(m\right)$ complexity with no sorting operation, which is much lower than the $O\left(mN\right)$ complexity of the GA method \cite{GA}.

\section{Conclusion}\label{section_conclusion}

In this letter, we propose two progressive rate-filling methods to realize the agile construction of MLC-PCM. Simulation results show that the proposed construction method is robust to diverse modulation and coding schemes.

\ifCLASSOPTIONcaptionsoff
  \newpage
\fi


\begin{thebibliography}{99}

\bibitem{arikan}
E. Ar{\i}kan, ``Channel polarization: A method for constructing capacity achieving codes for symmetric binary-input memoryless channels,'' \emph{IEEE Trans. Inf. Theory}, vol. 55, no. 7, pp. 3051--3073, Jul. 2009.

\bibitem{5G_NR_std}
\emph{Multiplexing and Channel Coding}, Release 15, 3GPP Standard TS 38.212, V15.8.0, Dec. 2019.


\bibitem{Polar_coded_modulation_seidl}
M. Seidl, A. Schenk, C. Stierstorfer and J. B. Huber, ``Polar-coded modulation,'' \emph{IEEE Trans. Commun.}, vol. 61, no. 10, pp. 4108--4119, Oct. 2013.

\bibitem{Polar_coded_MIMO_dai}
J. Dai, K. Niu and J. Lin, ``Polar-coded MIMO systems,'' \emph{IEEE Trans. Veh. Technol.}, vol. 67, no. 7, pp. 6170--6184, Jul. 2018.

\bibitem{MLC_wacha}
U. Wachsmann, R. F. H. Fischer, and J. B. Huber, ``Multilevel codes: theoretical concepts and practical design rules,'' \emph{IEEE Trans. Inf. Theory}, vol. 45, pp. 1361--1391, Jul. 1999.

\bibitem{Tse}
D. Tse and P. Viswanath, \emph{Fundamentals of Wireless Communications. Cambridge}, U.K.: Cambridge Univ. Press, 2005.

\bibitem{5G_NR_std_214}
\emph{Physical Layer Procedures for Data}, Release 16, 3GPP Standard TS 38.214, V16.0.0, Dec. 2019.

\bibitem{MIESM}
E. Dahlman, S. Parkvall, J. Skold, ``4G LTE/LTE-Advanced for Mobile Broadband,'' Elsevier, 2014.


\bibitem{DE_mori}
R. Mori and T. Tanaka, ``Performance of polar codes with the construction
using density evolution,'' \emph{IEEE Commun. Lett.}, vol. 13, no. 7, pp. 519--521,
Jul. 2009.

\bibitem{GA}
P. Trifonov, ``Efficient design and decoding of polar codes,'' \emph{IEEE Trans.
Commun.}, vol. 60, no. 11, pp. 3221--3227, Nov. 2012.

\bibitem{PW}
G. He \emph{et al.}, ``Beta-Expansion: A theoretical framework for fast and recursive construction of polar codes,'' in \emph{Proc. IEEE Global Communications Conference}, pp. 1--6, Singapore, Dec. 2017.


\bibitem{PCM_construction_wksp}
G. Bocherer, T. Prinz, P. Yuan and F. Steiner, ``Efficient polar code construction for higher-order modulation,'' in \emph{Proc. IEEE Wireless Communications and Networking Conference Workshops}, pp. 1--6, San Francisco, CA, 2017.


\bibitem{capacity_finite}
Y. Polyanskiy, H. V. Poor and S. Verdu, ``Channel coding rate in the finite blocklength regime,'' \emph{IEEE Trans. Inf. Theory}, vol. 56, no. 5, pp. 2307--2359, May 2010.

\bibitem{CASCL_niukai}
K. Niu and K. Chen, ``CRC-aided decoding of polar codes,'' \emph{IEEE Commun. Lett.}, vol. 16, no. 10, pp. 1668--1671, Oct. 2012.

%\bibitem{howtoconstruct_tal}
%I. Tal and A. Vardy, ``How to Construct Polar Codes,'' \emph{IEEE Trans. Inf. Theory}, vol. 59, no. 10, pp. 6562--6582, Oct. 2013.

\bibitem{rateless_libin}
B. Li, D. Tse, K. Chen and H. Shen, ``Capacity-achieving rateless polar codes,'' in \emph{Proc. IEEE Int. Symp. Inform. Theory (ISIT)}, pp. 46--50, Barcelona, 2016.

\bibitem{linshu}
S. Lin and D. J. Costello, ``Error Control Coding (2nd ed.),'' Prentice-Hall, Inc., 2004.

\end{thebibliography}
\end{document}